\documentclass[conference,letterpaper]{IEEEtran}
\usepackage[utf8]{inputenc} 
\usepackage[T1]{fontenc}
\usepackage{url}
\usepackage{ifthen}
\usepackage{cite}
\usepackage{bm}
\usepackage[cmex10]{amsmath} 
\usepackage{cite}
\usepackage{amssymb,amsfonts}
\usepackage{algorithmic}
\usepackage{graphicx}
\usepackage{textcomp}
\usepackage{xcolor}
\usepackage{epstopdf}
\usepackage{diagbox}
\usepackage{amsthm}
\usepackage{float}
\usepackage{subfigure}
\newtheorem{Le}{Lemma}
\newtheorem{theo}{Theorem} 
\newtheorem*{con}{Conjecture}

\begin{document}
\title{The Complete Affine Automorphism Group of Polar Codes} 

 \author{%
   \IEEEauthorblockN{Yuan Li\IEEEauthorrefmark{1}\IEEEauthorrefmark{2}\IEEEauthorrefmark{3},
                     Huazi Zhang\IEEEauthorrefmark{1},
                     Rong Li\IEEEauthorrefmark{1},
                     Jun Wang\IEEEauthorrefmark{1},
                     Wen Tong\IEEEauthorrefmark{1},
                     Guiying Yan\IEEEauthorrefmark{2}\IEEEauthorrefmark{3},
                     and Zhiming Ma\IEEEauthorrefmark{2}\IEEEauthorrefmark{3}}
   \IEEEauthorblockA{\IEEEauthorrefmark{1}%
                     Huawei Technologies Co. Ltd.}
  \IEEEauthorblockA{\IEEEauthorrefmark{2}%
                     University of Chinese Academy of Sciences}
   \IEEEauthorblockA{\IEEEauthorrefmark{3}%
                     Academy of Mathematics and Systems Science, CAS }
    Email: liyuan181@mails.ucas.ac.cn, \{zhanghuazi, lirongone.li, justin.wangjun, tongwen\}@huawei.com,\\
           yangy@amss.ac.cn, mazm@amt.ac.cn 

 }

\maketitle

\begin{abstract}
Recently, a permutation-based successive cancellation (PSC) decoding framework for polar codes attaches much attention. It decodes several permuted codewords with independent successive cancellation (SC) decoders. Its latency thus can be reduced to that of SC decoding. However, the PSC framework is ineffective for permutations falling into the lower-triangular affine (LTA) automorphism group, as they are  invariant under SC decoding. As such, a larger block lower-triangular affine (BLTA) group that contains SC-variant permutations was discovered for decreasing polar codes. But it was unknown whether BLTA equals the complete automorphism group. In this paper, we prove that BLTA equals the complete automorphisms of decreasing polar codes that can be formulated as affine trasformations.
\end{abstract}

\section{Introduction}
Polar codes \cite{b1}, invented by Ar{\i}kan, are a great break through in coding theory. As code length $N = 2^n$ approaches infinity, the synthesized channels become either noiseless or pure-noise, and the fraction of the noiseless channels approaches channel capacity. Thanks to channel polarization, efficient  SC decoding algorithm can be implemented with a complexity of $O(NlogN)$. However, the performance of polar codes under SC decoding is poor at short to moderate block lengths.

To boost finited-length performance, a successive cancellation list (SCL) decoding algorithm was proposed \cite{b2}. As list size $L$ increases, the performance of SCL decoding approaches that of maximum-likehood (ML) decoding. Accordingly, code construction is optimized for SCL decoding, e.g., CRC-aided (CA) \cite{b3} and parity-check (PC) \cite{b34}\cite{b30} polar codes. But in practice, a majority of SCL decoding complexity and latency is induced by path management, i.e., sorting and pruing paths according to path metric (PM). Recently, a PSC decoding framework \cite{b9}\cite{b26}\cite{b23} propose to decode $L$ permuted instances of the received codeword, and recover the most likely one in the end. In contrast to SCL decoding, these instances are independently decoded and do not require path management. Apparently, PSC decoding requires the permutations to be SC-variant. That is, instance SC decoders output distinct decoding results to achieve diversity gain. Current PSC decoders include stage permutation list decoding \cite{b9} that exploits stage permutations \cite{b20} and automorphism ensemble (AE) decoding \cite{b26}\cite{b23} that exploits the rich permutations found in polar automorphism groups. For a decreasing polar codes, there may not be enough SC-variant automorphisms available. Stage permutations can be included, although some of them do not fall into automorphism group. In \cite{b6} \cite{b7}, stage permutations are used to reduce the decoding complexity for RM codes. For polar codes, permutation decoding achieves similar performance to SCL in some cases with belief propagation (BP) \cite{b24} and SC \cite{b9} as instance decoders. 

The study of polar automorphism is inspired by \cite{b4}, where Reed-Muller (RM) codes are represented by monomials. The automorphism group of RM codes is shown to be affine transformation group of order $n$, denoted by $GA(n)$.  In \cite{b21}, polar codes  with partial order\cite{b31} are viewed as decreasing monomial codes, whose automorphism group includes the aforementioned LTA group. As an application, $\frac{N}{4}$-cyclic shift is proposed for implicit timing indication in Physical Broadcasting Channel (PBCH) \cite{b5}. In \cite{b23}, Geiselhart \emph{et al}. proposed an efficient algorithm to find permutations defined in a larger-than-LTA group called BLTA. The BLTA group is shown to be a subgroup of automorphism group and the authors further conjecture that the BLTA group is equal to polar automorphism group. Aiming at better PSC decoding performance, automorphisms in the upper-triangular linear (UTA) group are designed in \cite{b25}. A brief overview of the related works is illustrated in Fig. \ref{time}:
\begin{figure}[htbp]
	\centerline{\includegraphics[width=0.5\textwidth]{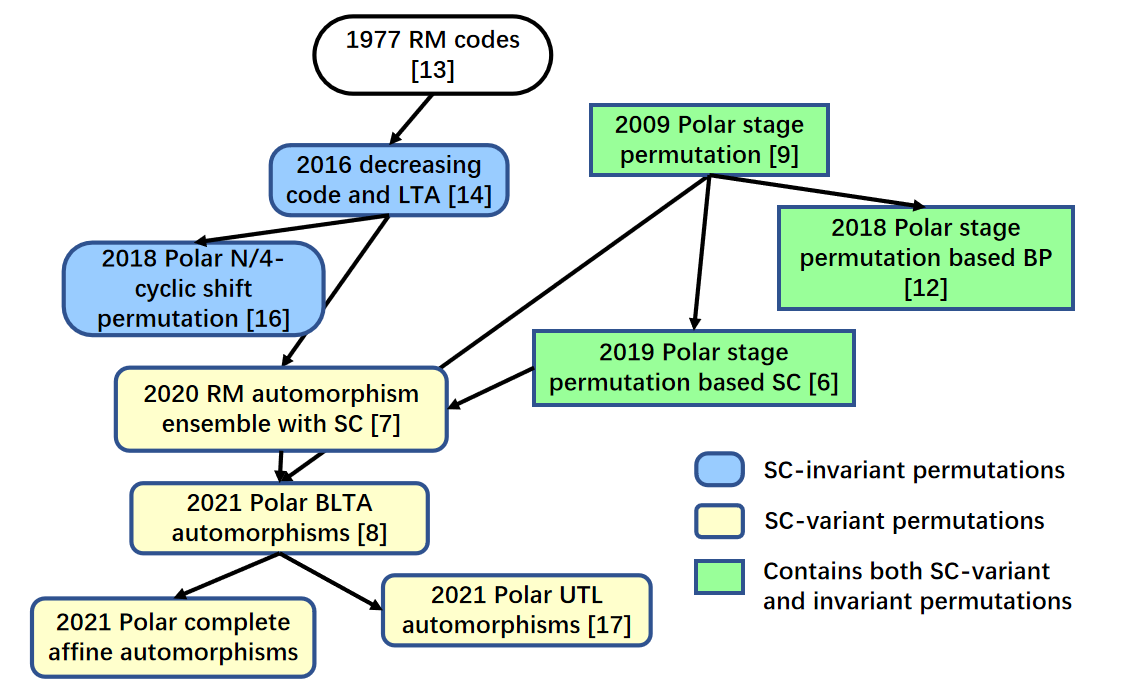}}
	\caption{Related works on automorphism group of RM and Polar codes}
	\label{time}
\end{figure}

In this paper, we prove that for decreasing codes, permutations in  BLTA are the complete affine transformation automorphisms, a conjecture that is a bit more constrained than that in \cite{b23}. This paper is organized as follows. In section II, we review the background of polar code automorphism groups. In section III we provide the proof. Finally we draw conclusions in section IV.

\section{Background}

\subsection{Polar Codes as Monomial Codes}

Given a B-DMC $W: \{ 0,1 \} \rightarrow \mathcal{Y}$,  the channel transition
probabilities are defined as $W(y|x)$, where $y \in \mathcal{Y} , x \in \{ 0,1 \}$. $W$ is said to be symmetric if there is a permutation $\pi$, such that $ \forall$ $y \in \mathcal{Y}$, $W(y|1)=W(\pi(y)|0)$ and $\pi^2 = id$.

Then the symmetric capacity and the Bhattacharyya parameter of  $W$ are defined as
\begin{equation*}
I(W) \triangleq \sum_{y \in \mathcal{Y}} \sum_{x \in \mathcal{X}} \frac{1}{2} W(y \mid x) \log \frac{W(y \mid x)}{\frac{1}{2} W(y \mid 0)+\frac{1}{2} W(y \mid 1)}
\end{equation*}
and
\begin{equation*}
Z(W) \triangleq \sum_{y \in \mathcal{Y}} \sqrt{W(y \mid 0) W(y \mid 1)}
\end{equation*}

Let
$
F=\left[\begin{array}{ll}
1 & 0 \\
1 & 1
\end{array}\right]
$, $N=2^n$, and $H_N=F^{\otimes n}$. Starting from $N = 2^n$ independent channels $W$, we obtain $N$ polarized channels $W_N^{(i)}$, after channel combining and splitting operations \cite{b1}, where
\begin{equation*}
W_{N}\left(y_{1}^{N}|u_1^N \right) \triangleq W^N \left(y_1^N|u_1^NH_N\right)
\end{equation*}
\begin{equation*}
W_{N}^{(i)}\left(y_{1}^{N}, u_{1}^{i-1} \mid u_{i}\right) \triangleq \sum_{u_{i+1}^{N} \in \mathcal{X}^{N-i}} \frac{1}{2^{N-1}} W_{N}\left(y_{1}^{N} \mid u_{1}^{N}\right)
\end{equation*}

Polar codes can be constructed by selecting the indices of $K$ information sub-channels, denoted by the information set $\mathcal{A} = \left\{ I_1,I_2,\dots,I_K \right\}$. The optimal sub-channel selection criterion for SC decoding is reliability, i.e., selecting the $K$ most reliable sub-channel as information set. Density evolution (DE) algorithm\cite{b10}, Gaussian approximation (GA) algorithm\cite{b11} and the channel-independent PW construction method\cite{b12} are efficient methods to find reliable sub-channels.

In particular, polar codes can be expressed as monomial codes\cite{b21}. From this point of view, each synthetic channel can be represented by a monomial with $n$ binary variable \{$x_i$\}, $0 \leq i \leq n-1$, and the monomial set can be denoted by
$$\mathcal{M}_n \overset{def}{=} \{ x_{0}^{g_{0}} x_{1}^{g_{1}}\ldots x_{n-1}^{g_{n-1}} |({g_{0}},{g_{1}},\dots,{g_{n-1}}) \in \mathbf F_2^n \}$$
For instance, $f= x_{i_1}x_{i_2}\ldots x_{i_s}$. The degree of $f$ is $s$, denoted by $deg(f)$. Each row of the $H_N$ can be expressed as a monomial, and thus we can use a subset of $\mathcal{M}_n$ to denote the polar code. For convenience, we denote the information set by $M$, and the polar code spanned by $M$ as $C(M)$.

\subsection{Decreasing Monomial Codes}
It was presented in \cite{b21} and \cite{b31} that
the reliablity of synthetic channels follows a partial order ``$\preceq$''. If $f,g\in \mathcal{M}_n$, $g \preceq f$ means $g$ is universally more reliable than $f$. 
For monomials of the same degree, partial order is defined as 
$$x_{i_1}\ldots x_{i_r} \preceq x_{j_1}\ldots x_{j_r} \iff i_k \leq j_k, \ 1 \leq k \leq r$$ 
and for monomials of different degree
$$g\preceq f \iff \exists f^*\mid f, \  deg(f^*)=deg(g), \  and \ g\preceq f^*$$.

A decreasing monomial code $C(M)$ is a monomial code satisfying partial order. i.e.
$$
\forall f\in M \ and \ g\in \mathcal{M}_n,\ if\ g\preceq f\Rightarrow g\in M 
$$

In practice, many polar codes can be regarded as decreasing monomial codes, i.e., we can find the ``largest" monomials $M_{min}$ as generators \cite{b23}, and the information set $M$ can be defined as: 
$$
M=\bigcup\limits_{f\in M_{min}}\{g\in \mathcal{M}_n \mid g\preceq f \}
$$

\subsection{Automorphisms of Decreasing Monomial codes}
The automorphism group $Aut(M)$ of $C(M)$ is defined as the set of permutations $\pi \in S_N$, where $N$ is the code length of $C(M)$, $\pi \in Aut(M)$ if and only if
$$ \pi(c) \in C(M), \ \forall \ c \in C(M)$$
where
$$\pi(c)_{i} = c_{\pi(i)}, \ 0 \leq i \leq N-1$$

It's well known that the automorphism group of Reed-Muller codes of length $N=2^n$ is given by the affine transformation group $GA(n)$ \cite{b4}, that is 

\begin{equation}\label{key1}
X\stackrel{(\bm{A},\bm{b}) \in GA(n)}{\longrightarrow} Y=\bm{A}X+ \bm{b}
\end{equation}
with $X,Y\in \mathbb{F}_2^n $ and $\bm{A}$ an $n\times n$ binary invertible matrix plus a binary column vector $\bm{b}$  of length $n$. In \cite{b21}, it is shown that the automorphism group of a decreasing monomial code contains at least $LTA(n)$, where $\bm{A}$ is a lower triangular matrix. Then, \cite{b23} prove that  BLTA is a larger automorphism subgroup containing the LTA. The BLTA is in the form of (\ref{key1}), where $\bm{A}$ is shown in Fig. \ref{fig1}:
\begin{figure}[htbp]
	\centerline{\includegraphics[width=0.3\textwidth]{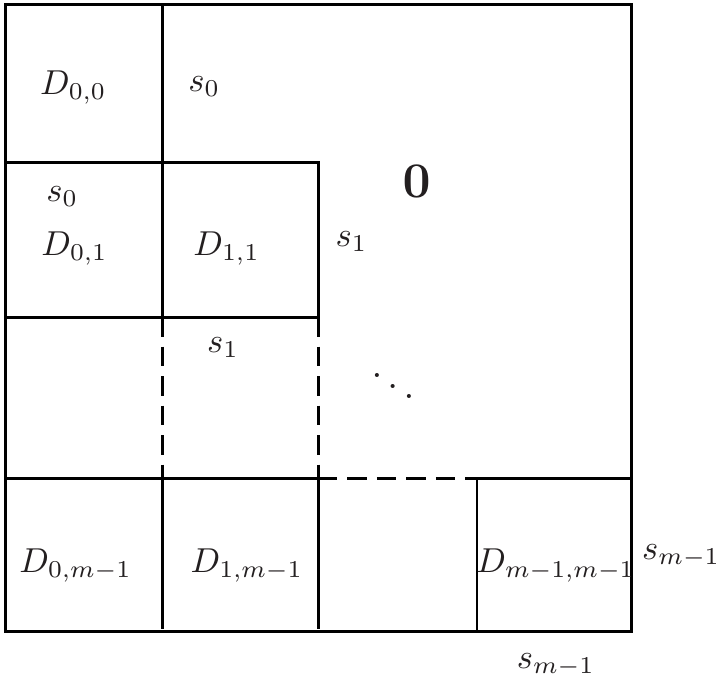}}
	\caption{The block lower triangular matrix $\bm{A}$ of the affine transformation \cite{b23} }
	\label{fig1}
\end{figure}
$\bm{D}_{i,i}$ in the diagonal are invertible binary random matrixs of size $s_i\times s_i$ and  $\bm{D}_{i,j\neq i}$ are binary random matrices. 

As indicated in \cite{b25}, not all automorphisms of decreasing monomial codes can be represented by affine transformations. Up to now, we do not have a unified framework to analyze the automorphisms which can not be viewed as affine transformations. In the paper, we mainly focus on the affine transformation automorphisms of decreasing polar codes. 

\subsection{Permutation-based SC Decoding}
The PSC decoding framework \cite{b9}\cite{b26}\cite{b23} is shown in Fig. \ref{fig2}. $L$ different permutations are applied to the received vector $y$. Each is decoded by an SC-based decoder to obtain a permuted codeword $x'_i$, which is deinterleaverd to $x_i$. Finally, the most likely candidate codeword is selected as the decoding output.

The permutations used in PSC decoding need to be carefully selected. It is proved in \cite{b26} that permutations from LTA are SC-invariant, i.e. $SC(\pi(L_{ch}))=\pi(SC(L_{ch}))$. This renders PSC decoding useless because all SC decoder instances output the same codeword. In order to improve PSC decoding performance, we need to find more SC-variant permutations. They can be obtained either by permuting the stages of factors graph \cite{b9} \cite{b20}, or by exploring a larger automorphism group \cite{b26} \cite{b23} \cite{b25}.
\begin{figure}[htbp]
	\centerline{\includegraphics[width=0.5\textwidth]{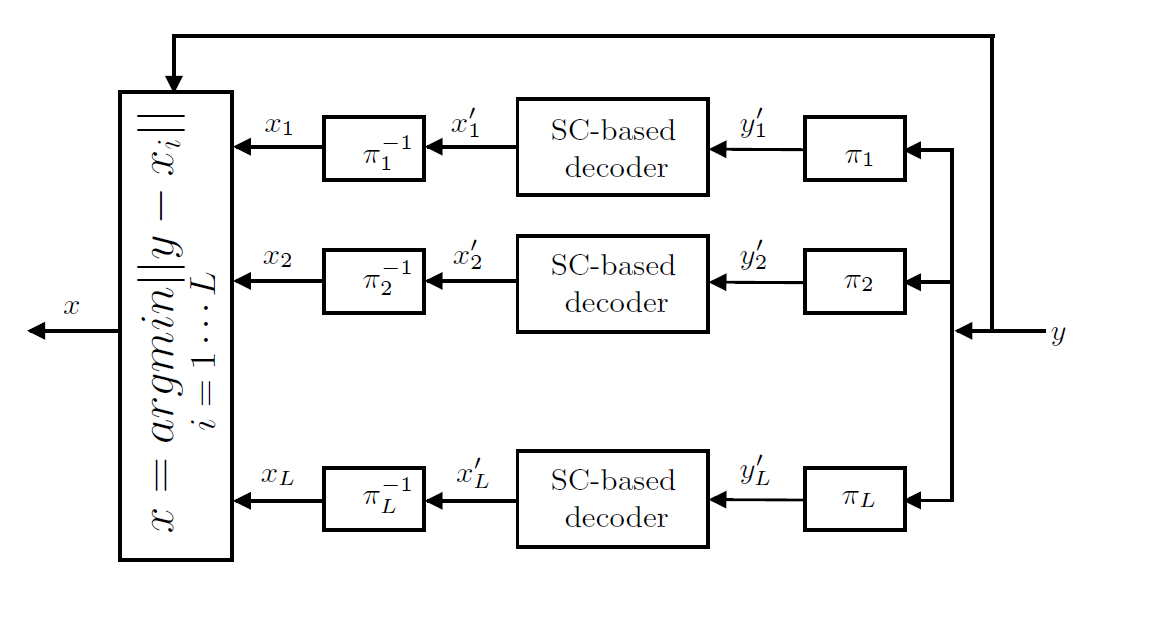}}
	\caption{The PSC decoding framework \cite{b9} \cite{b26} \cite{b23} }
	\label{fig2}
\end{figure}

\section{Analysis on Automorphisms of Decreasing Monomial Codes}
In this section, we prove that for decreasing codes, all the automorphisms that can be expressed as affine transformations are equal to BLTA. 

\subsection{Notations and Definitions}
Let $[i,j] \triangleq [i,i+1,\dots,j]$. Let $Aut(M)$ be the automorphism group of $C(M)$. $BLTA(s,n)$ is the same as defined in \cite{b23}.

We define the affine automorphism group to be the subgroup of $Aut(M)$, which includes all the automorphisms that can be expressed as affine transformations. And the affine automorphism group of monomial codes $M$ is abbreviated as $A\text{-}Aut(M)$.

Let $S_{ \{ i_1,\dots,i_k \}}$ denote the permutation group of the set $\{i_1,\dots,i_k\}$, in particular, $S_n$ is the permutation group of the set $[0,n-1]$. A permutation $\pi \in S_n$ can be expressed as $(\bm{A}_{\pi},\bm{0})$, where $\bm{A}_{\pi}$ is a $n \times n$ permutation matrix. 

If $(\bm{A},\bm{0}) \in A \text{-}Aut(M)$, we abbreviate it to $\bm{A} \in A \text{-}Aut(M)$. When $\bm{A} = \bm{A}_{\pi}$, we further abbreviate it to $\pi \in A \text{-}Aut(M)$.

For a $n\times n$ matrix $\bm{A}$, let $\bm{A}_{\{i_1,\dots,i_s\},\{j_1,\dots,j_t\}}$ be the $s \times t$ corresponding submatrix of $\bm{A}$. Where $i_1,\dots,i_s$ are distinct integers and $j_1,\dots,j_t$ are distinct integers as well.

\subsection{Analysis on Automorphisms}
In \cite{b23}, the author proved that $BLTA(s,n) \subseteq Aut(M)$, and further conjectured that the equality holds. That is, BLTA group is equal to polar automorphism group. However, the conjecture seems a bit too aggressive as shown in \cite{b25}, where numerical experiments found some permutations in $Aut(M)$ are outside BLTA. Therefore, $BLTA(s,n) \neq Aut(M)$.

However, the conjecture would be accurate if we focus on affine automorphisms. In this paper, we prove the following conjecture for all decreasing monomial codes, in which decreasing polar codes and RM codes are special cases. 
\begin{con}
$$BLTA(s,n) = A\text{-}Aut(M)$$    
\end{con}  

According to \cite{b23}, we only need to prove that if $(\bm{A},\bm{b}) \in A$-$Aut(M)$ and $a_{i,j} = 1, i < j$, then $(i,j) \in A\text{-}Aut(M)$. It amounts to proving the following Theorem.

\begin{theo}
Let $C(M)$ be a decreasing monomial code in $n$ variables with information set $M$. If $\exists \ (\bm{A},\bm{b}) \in A \text{-}Aut(M), a_{i,j}=1, i < j$, then 
$$\pi = (i,j) \in A \text{-}Aut(M) $$
according to \cite{b23},\\
\begin{center}
$BLTA(s,n)=A$-$Aut(M)$
\end{center}
\end{theo}

In order to prove \textbf{Theorem 1}, we only need to prove the following Theorem.

\begin{theo}
Let $C(M)$ be a decreasing monomial code in $n$ variables with information set $M$. If \ $\exists \ \bm{A} \in A \text{-}Aut(M), a_{i,i+1}=1$, then 
$$\pi = (i,i+1) \in A \text{-}Aut(M) $$
\end{theo}
\begin{proof} \textbf{Theorem 2} $\Rightarrow$ \textbf{Theorem 1}

If $\exists \ (\bm{A},\bm{b}) \in A \text{-}Aut(M), a_{i,j}=1, i < j$, then we can prove that $\bm{A} \in A \text{-}Aut(M)$.

Because $C(M)$ is a decreasing monomial code, $(\bm{I_n},\bm{b}) \in LTA(n) \subseteq A \text{-}Aut(M)$, where $\bm{I_n}$ is the identity matrix of order $n$. By group's closure property, we have $(\bm{I_n},\bm{b}) \circ (\bm{A},\bm{b})=(\bm{A},\bm{0}) \in A \text{-}Aut(M)$.

If \ $\exists$ $a_{k,k+1}=0$, $i \leq k \leq j-1$, because $\bm{L_1AL_2} \in A \text{-}Aut(M)$, where $\bm{L_1}$, $\bm{L_2}$ are invertible lower triangular matrices. That means we can add the latter columns of $\bm{A}$ to the preceding columns, or add the top rows of $\bm{A}$ to the bottom rows to get a new invertible matrix $\bm{B}$, and $\bm{B} \in A \text{-}Aut(M)$. If $a_{i,k+1}=1$, add the $i$-$th$ row of $\bm{A}$ to the $k$-$th$ row, and if $a_{i,k+1}=0$, add the $j$-$th$ column of $\bm{A}$ to the $(k+1)$-$th$ column, then add the $i$-$th$ row of $\bm{A}$ to the $k$-$th$ row. As a result, we get a new invertible matrix $\bm{B} \in A \text{-}Aut(M)$,
 where $b_{k,k+1}=1, \forall \ i \leq k < j$.

Accoding to \textbf{Theorem 2}, $(i,i+1),\dots,(j-1,j) \in A \text{-}Aut(M)$, so 
$$(j-1,j)\circ \dots \circ(i,i+1) \in A \text{-}Aut(M) $$
By \cite{b23}, we conclude that
$$\pi = (i,j) \in A \text{-}Aut(M)$$
\end{proof}

To prove \textbf{Theorem 2}, we need the following lemmas

\begin{Le}
Denote an affine transformation by $x_{n_1} \dots x_{n_r}\stackrel{\bm{A}}{\longrightarrow}y_{n_1} \dots y_{n_r}$, $y_m = \sum\limits_{i=0}^{n-1}a_{m,i}x_i$. Then $\forall$ $0 \leq j_1,\dots, j_r \leq n-1$, where $j_1,\dots, j_r$ are distinct integers, the coefficient of $x_{j_1} \dots x_{j_r}$ is non-zero in the expansion of $y_{n_1} \dots y_{n_r}$ iff
$$det(\bm{A}_{\{n_1,\dots, n_r\},\{j_1,\dots,j_r\}}) \neq 0$$
\end{Le}
\begin{proof}
\begin{align*}
& \ \ \ \ y_{n_1} \dots y_{n_r}\\
&=\sum_{m_1=0}^{n-1} \dots \sum_{m_r=0}^{n-1} a_{n_1,m_1} \dots a_{n_r,m_r} x_{m_1} \dots x_{m_r}\\
&=\sum_{\sigma \in S_{\{n_1,\dots,n_r\}}}a_{n_1,\sigma(j_1)} \dots a_{n_r,\sigma(j_r)}x_{j_1} \dots x_{j_r} + R\\
&=\sum_{\sigma \in S_{\{n_1,\dots,n_r\}}}(-1)^{\sigma}a_{n_1,\sigma(j_1)} \dots a_{n_r,\sigma(j_r)}x_{j_1} \dots x_{j_r} + R\\
&=det(\bm{A}_{\{n_1,\dots, n_r\},\{j_1,\dots,j_r\}})x_{j_1} \dots x_{j_r} + R
\end{align*}
The third equality is due to $-1 = 1 \ (mod \ 2)$, and $R$ contains all the terms except $x_{j_1} \dots x_{j_r}$.
\end{proof}

\begin{Le}
Let $\bm{D}$ be a $P \times Q$ matrix, rank($\bm{D}$) = t, if $det(\bm{D}_{\{p_1,\dots,p_r\},\{q_1,\dots,q_r\}}) \neq 0$, and $r < t$, then there exists a submatrix $\bm{D}_{\{p_1,\dots,p_r,p_{r+1},\dots,p_t\},\{q_1,\dots,q_r,q_{r+1},\dots,q_t\}}$ containing $\bm{D}_{\{p_1,\dots,p_r\},\{q_1,\dots,q_r\}}$, and 
$$rank(\bm{D}_{\{p_1,\dots,p_r,p_{r+1},\dots,p_t\},\{q_1,\dots,q_r,q_{r+1},\dots,q_t\}}) = t$$
\end{Le}
\begin{proof}
Because 
$$rank\left(\bm{D}_{\{p_1,\dots,p_r\},\{[0,Q-1]\}}\right)=r, \ rank\left(\bm{D}\right)=t$$

We can extend $\{p_1,\dots,p_r\}$ to $\{p_1,\dots,p_r,p_{r+1},\dots,p_t\}$, such that $rank\left(\bm{D}_{\{p_1,\dots,p_r,p_{r+1},\dots,p_t\},\{[0,Q-1]\}}\right)=t$.

For the same reason, we can extend $\{q_1,\dots,q_r\}$ to $\{q_1,\dots,q_r,q_{r+1},\dots,q_t\}$, such that 
$$rank\left(\bm{D}_{\{p_1,\dots,p_r,p_{r+1},\dots,p_t\},\{q_1,\dots,q_r,q_{r+1},\dots,q_t\}}\right)=t$$
\end{proof}

\begin{Le}
Let $\{\bm{a}_1,\dots,\bm{a}_m\}$ be linearly independent column vectors, if $\{\bm{a}_1,\dots,\bm{a}_{m-1},\bm{a}_n\}$ are linearly dependent, then $\{\bm{a}_1,\dots,\bm{a}_{m-1},\bm{a}_n+\bm{a}_m\}$ are linearly independent.
\end{Le}
\begin{proof}
Because $\{\bm{a}_1,\dots,\bm{a}_{m-1},\bm{a}_n\}$ are linearly dependent, and $\{\bm{a}_1,\dots,\bm{a}_{m-1}\}$ are linearly independent, we have 
$$\bm{a}_n = \sum\limits_{k=1}^{m-1}c_k \bm{a}_k, \ c_k \in \{ 0,1 \}$$
If 
$$\bm{a}_n + \bm{a}_m  = \sum\limits_{k=1}^{m-1}c'_k \bm{a}_k, \ c'_k \in \{ 0,1 \}$$ 
then 
$$\bm{a}_m  = \sum\limits_{k=1}^{m-1}(c'_k-c_k) \bm{a}_k$$
This contradicts that $\{\bm{a}_1,\dots,\bm{a}_m\}$ are linearly independent.

\end{proof}

\begin{proof}[Proof of \textbf{Theorem 2}]
we need to prove that 

$\forall$ $x_{i_1} \dots x_{i_{r+1}} \in M$, then $y_{i_1} \dots y_{i_{r+1}} \in M$, where 
\begin{equation*}    y_j =
 \begin{cases}
    x_{i+1} & j = i\\
    x_{i} & j = i+1\\
    x_j   & otherwise
\end{cases}
\end{equation*}
and $i_1 < i_2 \dots < i_{r+1}$.

We take a divide-and-conquer approach. If $i,i+1 \notin \{i_1,\dots,i_{r+1}\}$ or $i,i+1 \in \{i_1,\dots,i_{r+1}\}$, the proof is straightforward as $y_{i_1} \dots y_{i_{r+1}}=x_{i_1} \dots x_{i_{r+1}} \in M$. 

If $i \notin \{i_1,\dots,i_{r+1}\}$ and $i+1 \in \{i_1,\dots,i_{r+1}\}$, the proof can be obtained by the "decreasing" property. Assuming $i_m = i+1,1 \leq m \leq r+1$, then 
\begin{align*}
&\ \ \ \ y_{i_1} \dots y_{i_{r+1}}\\
&=x_{i_1} \dots x_{i_{m-1}} x_i x_{i_{m+1}} \dots x_{i_{r+1}} \\
&\preceq x_{i_1} \dots x_{i_{m-1}} x_{i+1} x_{i_{m+1}} \dots x_{i_{r+1}} \\
&= x_{i_1} \dots x_{i_{r+1}} 
\end{align*}
Because $C(M)$ is a decreasing monomial code, then $y_{i_1} \dots y_{i_r+1} \in M$.

What remains to be proved is the most tricky case where $i \in \{i_1,\dots,i_{r+1}\}$ and $i+1 \notin \{i_1,\dots,i_{r+1}\}$. We further divide it to the following three cases.

\emph{Case 1}: $i = i_{r+1}$

Consider the submatrix $\bm{A}_{\{[0,i_1],i\},\{[i_1,n-1]\}}$, because
\begin{align*}
&rank\left(\bm{A}_{\{[0,i_1],i\},\{[i_1,n-1]\}}\right) + rank\left(\bm{A}_{\{[0,i_1],i\},\{[0,i_1-1]\}}\right)\\
&\geq rank\left(\bm{A}_{\{[0,i_1],i\},\{[0,n-1]\}}\right) = i_1+2
\end{align*}
and $rank\left(\bm{A}_{\{[0,i_1],i\},\{[0,i_1-1]\}}\right) \leq i_1$, we obtain that
$$rank\left(\bm{A}_{\{[0,i_1],i\},\{[i_1,n-1]\}}\right) \geq 2$$

According to \textbf{Lemma 2}, $\exists \ 0 \leq s_1 \leq i_1$, $i_1 \leq t_1 \leq n-1$, $i \neq s_1$ and $i+1 \neq t_1$, such that $det\left(\bm{A}_{\{ s_1,i \},\{t_1,i+1\}}\right) \neq 0$, if $det\left(\bm{A}_{\{s_1,i \},\{ i_1,i+1\}}\right) \neq 0$, define $\bm{A}^{(1)}=\bm{A}$, otherwise, add the $t_1$-$th$ column of $\bm{A}$ to the $i_1$-$th$ column and denote the new matrix by $\bm{A}^{(1)}$, according to \textbf{Lemma 3} $det\left(\bm{A}^{(1)}_{\{s_1,i \},\{ i_1,i+1\}}\right) \neq 0$ and $\bm{A}^{(1)} \in A \text{-}Aut(M)$. An example is shown in Fig. \ref{fig3}.

\begin{figure}[htbp]
	\centerline{\includegraphics[width=0.4\textwidth]{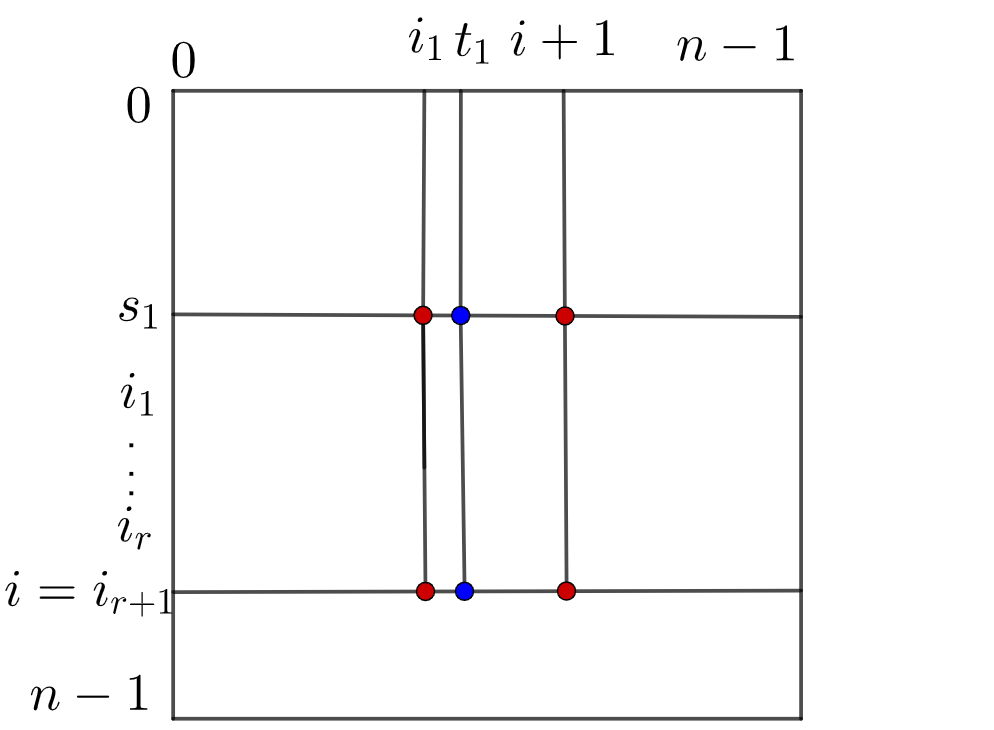}}
	\caption {Operations on the BLTA matrix $\bm{A}$ for \emph{case 1}}
	\label{fig3}
\end{figure}

Suppose we have $\bm{A}^{(k)} \in A \text{-}Aut(M)$, $1 \leq k < r$, and $det\left(\bm{A}^{(k)}_{\{s_1,\dots,s_k,i \},\{ i_1,\dots,i_k,i+1 \}}\right) \neq 0$, $s_m \leq i_m$, $1 \leq m \leq k$.
Because
\begin{align*}
& \ \ \ \ rank\left(\bm{A}^{(k)}_{\{[0,i_{k+1}],i\},\{i_1,\dots,i_{k},[i_{k+1},n-1]\}}\right)\\
&\geq i_{k+1}+2 - (i_{k+1}-k)=k+2
\end{align*}
According to \textbf{Lemma 2}, $\exists \ 0 \leq s_{k+1} \leq i_{k+1}$, $i_{k+1} \leq t_{k+1} \leq n-1$, and $s_1,\dots,s_{k+1},i$ is a set of distinct indices, and $i_1,\dots,i_k,t_{k+1},i+1$ is another set of distinct indices, such that $det\left(\bm{A}^{(k)}_{\{s_1,\dots,s_{k+1},i\},\{i_1,\dots,i_k,t_{k+1},i+1\}}\right) \neq 0$, again if $det\left(\bm{A}^{(k)}_{\{s_1,\dots,s_{k+1},i\},\{i_1,\dots,i_k,i_{k+1},i+1\}}\right) \neq 0$, define $\bm{A}^{(k+1)}=\bm{A}^{(k)}$, otherwise, add the $t_{k+1}$-$th$ column of $\bm{A}^{(k)}$ to the $i_{k+1}$-$th$ column and denote the new matrix by $\bm{A}^{(k+1)}$.

Finally, we obtain $A^{(k+1)} \in A \text{-}Aut(M)$, and $det\left(\bm{A}^{(k+1)}_{\{s_1,\dots,s_{k+1},i\},\{i_1,\dots,i_{k+1},i+1\}}\right) \neq 0$. 

When $k = r-1$, we have $A^{(r)} \in A \text{-}Aut(M)$, and $det\left(\bm{A}^{(r)}_{\{s_1,\dots,s_r,i\},\{i_1,\dots,i_r,i+1\}}\right) \neq 0$. Because $s_m \leq i_m$, $1 \leq m \leq r$,
$$x_{s_1} \dots x_{s_r} x_i \preceq x_{i_1} \dots x_{i_r} x_i$$
By definition of dereasing monomial codes, we know $x_{s_1} \dots x_{s_r} x_i \in A\text{-}Aut(M)$.

Denote an affine transformation by $ x_{s_1} \dots x_{s_r} x_i \stackrel{\bm{A}^{(r)}}{\longrightarrow} y_{s_1} \dots y_{s_r} y_i$, because $det\left(\bm{A}^{(r)}_{\{s_1,\dots,s_r,i\},\{i_1,\dots,i_r,i+1\}}\right) \neq 0$ and due to \textbf{Lemma 1}, we have
$$ y_{s_1} \dots y_{s_r} y_i= x_{i_1}\dots x_{i_r} x_{i+1}+ R$$
Therefore, $x_{i_1}\dots x_{i_r} x_{i+1}\in M$\cite{b25}.

\emph{Case 2}: $i = i_1$\\
Because
\begin{align*}
& \ \ \ \ rank \left(\bm{A}_{\{[0,i_2]\},\{i+1,[i_2,n-1]\}}\right)\\ 
&\geq i_2+1-(i_2-1) = 2
\end{align*}
According to \textbf{Lemma 2}, as in \emph{Case 1}, we can obtain a matrix $\bm{A}^{(1)}$, such that $det\left(\bm{A}^{(1)}_{\{i,s_2 \},\{ i+1,i_2\}}\right) \neq 0$ and $\bm{A}^{(1)} \in A \text{-}Aut(M)$. An example is shown in Fig. \ref{fig4}.
\begin{figure}[htbp]
	\centerline{\includegraphics[width=0.4\textwidth]{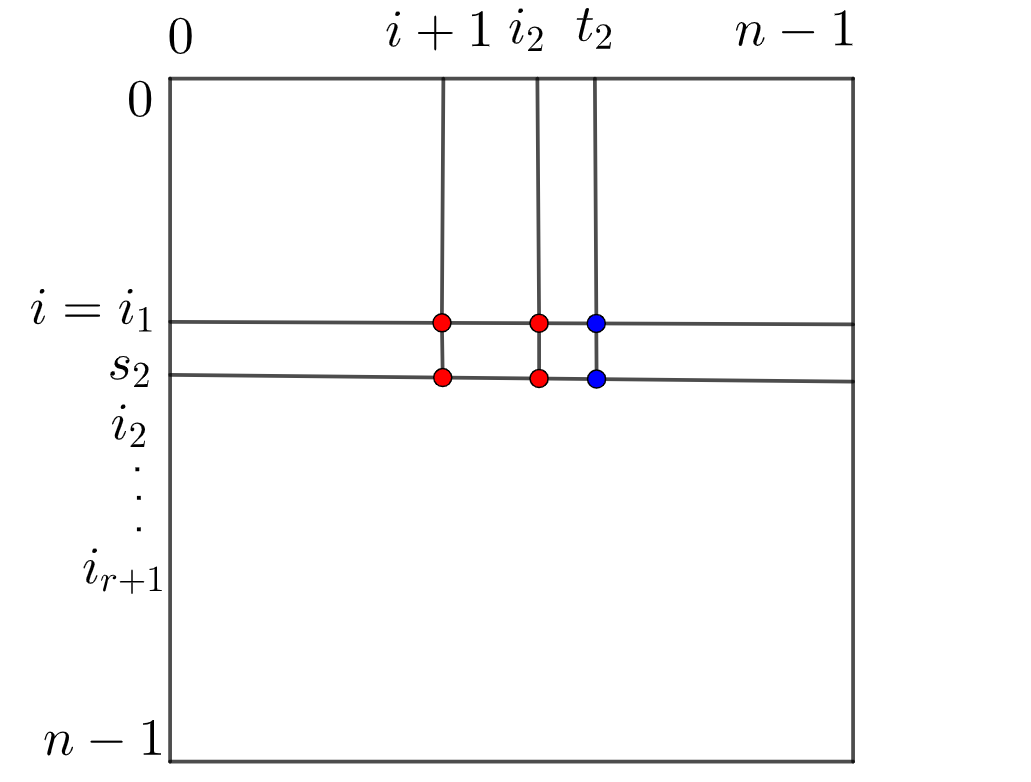}}
	\caption {Operations on the BLTA matrix $\bm{A}$ for \emph{case 2}}
	\label{fig4}
\end{figure}

Suppose we have $\bm{A}^{(k)} \in A \text{-}Aut(M)$, $1 \leq k < r$, and $det\left(\bm{A}^{(k)}_{\{i,s_2,\dots,s_{k+1} \},\{ i+1,i_2,\dots,i_{k+1} \}}\right) \neq 0$, $s_m \leq i_m$, $2 \leq m \leq k+1$.
Because
\begin{align*}
& \ \ \ \ rank\left(\bm{A}^{(k)}_{\{[0,i_{k+2}]\},\{i+1,i_2,\dots,i_{k+1},[i_{k+2},n-1]\}}\right)\\
&\geq i_{k+2}+1 - (i_{k+2}-k-1)=k+2
\end{align*}
Again as in \emph{Case 1}, we can obtain a matrix $\bm{A}^{(k+1)} \in A \text{-}Aut(M)$, and $det\left(\bm{A}^{(k+1)}_{\{i,s_2,\dots,s_{k+2}\},\{i+1,i_2,\dots,i_{k+2}\}}\right) \neq 0$.

When $k = r-1$, we have $\bm{A}^{(r)} \in A \text{-}Aut(M)$, and $det\left(\bm{A}^{(r)}_{\{i,s_2,\dots,s_{r+1}\},\{i+1,i_2,\dots,i_{r+1}\}}\right) \neq 0$, $s_m \leq i_m$, $2 \leq m \leq r+1$. The rest of proof is the same as in \emph{Case 1}.

\emph{Case 3}: $i = i_m$, $1 < m < r+1$

According to \emph{Case 1}, we can obtain a matrix $\bm{A}^{(m-1)} \in A \text{-}Aut(M)$, 
$$det\left(\bm{A}^{(m-1)}_{\{s_1,\dots,s_{m-1},i\},\{i_1,\dots,i_{m-1},i+1 \}}\right) \neq 0$$
Then according to \emph{Case 2}, we can obtain a matrix $\bm{A}^{(r)} \in A \text{-}Aut(M)$, such that
\begin{align*}
&\ \ \ \ det\left(\bm{A}^{(r)}_{\{s_1,\dots,s_{m-1},i,s_{m+1}\dots s_{r+1}\},\{i_1,\dots,i_{m-1},i+1,i_{m+1},\dots i_{r+1} \}}\right)\\
&\neq 0
\end{align*}
The rest of proof is the same as in \emph{Case 1}. 
\end{proof}

\section{Conclusion}
In this paper, we prove the conjecture that BLTA is the complete affine automorphism group for decreasing monomial codes, including decreasing polar codes and RM codes. Our proof guarantees that all the automorphisms defined by affine transformation can be found for decreasing polar codes.

\section{Acknowledgement}
The authors thank Xianbin Wang for the fruitful discussions that inspired this work, and Zhipeng Gao for the valuable comments.



\begin{thebibliography}{9}

\bibitem{b1} E. Ar{\i}kan, "Channel Polarization: A Method for Constructing Capacity-Achieving Codes for Symmetric Binary-Input Memoryless Channels," in \emph{IEEE Transactions on Information Theory}, vol. 55, no. 7, pp. 3051-3073, Jul. 2009.
\bibitem{b2} I. Tal and A. Vardy, "List Decoding of Polar Codes," in \emph{IEEE Transactions on Information Theory}, vol. 61, no. 5, pp. 2213-2226, May 2015.
\bibitem{b3} K. Niu and K. Chen, "CRC-Aided Decoding of Polar Codes," in \emph{IEEE Communications Letters}, vol. 16, no. 10, pp. 1668-1671, Oct. 2012.
\bibitem{b34}P. Trifonov and V. Miloslavskaya, "Polar Subcodes," in \emph{IEEE Journal on Selected Areas in Communications}, vol. 34, no. 2, pp. 254-266, Feb. 2016.
\bibitem{b30} H. Zhang et al., "Parity-Check Polar Coding for 5G and Beyond," \emph{IEEE International Conference on Communications (ICC)}, Kansas City, MO, 2018, pp. 1-7.
\bibitem{b9} M. Kamenev, Y. Kameneva, O. Kurmaev and A. Maevskiy, "Permutation Decoding of Polar Codes,"  \emph{XVI International Symposium "Problems of Redundancy in Information and Control Systems" (REDUNDANCY)}, Moscow, Russia, 2019, pp. 1-6.
\bibitem{b26} M. Geiselhart, A. Elkelesh, M. Ebada, S. Cammerer, S.ten Brink. "Automorphism Ensemble Decoding of Reed-Muller Codes" 	arXiv:2012.07635.
\bibitem{b23} M. Geiselhart, A. Elkelesh, M. Ebada, S. Cammerer, S. ten Brink. "On the Automorphism Group of Polar Codes." arXiv preprint arXiv:2101.09679 (2021).
\bibitem{b20}N. Hussami, S. B. Korada and R. Urbanke, "Performance of polar codes for channel and source coding,"  \emph{IEEE International Symposium on Information Theory}, Seoul, Korea (South), 2009, pp. 1488-1492.
\bibitem{b6} M. Kamenev, Y. Kameneva, O. Kurmaev and A. Maevskiy, "A New Permutation Decoding Method for Reed-Muller Codes," \emph{IEEE International Symposium on Information Theory (ISIT)}, Paris, France, 2019, pp. 26-30.
\bibitem{b7} K. Ivanov and R. Urbanke, "Permutation-based Decoding of Reed-Muller Codes in Binary Erasure Channel," \emph{IEEE International Symposium on Information Theory (ISIT)}, Paris, France, 2019, pp. 21-25.
\bibitem{b24}  A. Elkelesh, M. Ebada, S. Cammerer, and S. ten Brink, “Belief propagation list decoding of polar codes,” \emph{IEEE Communications Letters}, vol. 22, no. 8, pp. 1536–1539, Aug. 2018.
\bibitem{b4} F. J. MacWilliams, N. J. A. Sloane, "The theory of error correcting codes," Elsevier, 1977.
\bibitem{b21} M. Bardet, V. Dragoi, A. Otmani and J. Tillich, "Algebraic properties of polar codes from a new polynomial formalism," \emph{IEEE International Symposium on Information Theory (ISIT)}, Barcelona, Spain, 2016, pp. 230-234
\bibitem{b31} C. Schürch, "A partial order for the synthesized channels of a polar code," \emph{IEEE International Symposium on Information Theory (ISIT)}, Barcelona, Spain, 2016, pp. 220-224.
\bibitem{b5} H. Luo et al., "Analysis and Application of Permuted Polar Codes," \emph{IEEE Global Communications Conference (GLOBECOM)}, Abu Dhabi, United Arab Emirates, 2018, pp. 1-5.
\bibitem{b25} C. Pillet, V. Bioglio, and I. Land. "Polar Codes for Automorphism Ensemble Decoding." arXiv preprint arXiv:2102.08250 (2021).
\bibitem{b10} R. Mori and T. Tanaka, "Performance of Polar Codes with the Construction using Density Evolution," in \emph{IEEE Communications Letters}, vol. 13, no. 7, pp. 519-521, Jul. 2009.
\bibitem{b11} P. Trifonov, "Efficient Design and Decoding of Polar Codes," in \emph{IEEE Transactions on Communications}, vol. 60, no. 11, pp. 3221-3227, Nov. 2012.
\bibitem{b12} G. He et al., "$\beta$-Expansion: A Theoretical Framework for Fast and Recursive Construction of Polar Codes," \emph{IEEE Global Communications Conference}, Singapore, 2017, pp. 1-6. 








\end{thebibliography}


\end{document}